\documentclass[10pt, conference, compsocconf]{IEEEtran}
\usepackage[french,linesnumbered,algoruled]{algorithm2e}
\renewcommand{\algorithmcfname}{ALGORITHM}
\ifCLASSINFOpdf
\else
\fi

\newtheorem{definition}{Definition}
\newtheorem{example}{Example}

\newtheorem{proposition}{Proposition}

\newtheorem{proof}{Proof}
\newtheorem{remark}{Remark}
\newtheorem{problem}{Problem}
\newtheorem{lemma}{Lemma}

\usepackage{longtable}
\usepackage{pstricks}
\usepackage{verbatim}
\usepackage{times}
\usepackage{helvet}
\usepackage{courier}
\usepackage{amssymb}
\usepackage{graphicx}
\usepackage{amsmath}

\begin{document}

\title{A scalable mining of frequent quadratic concepts in d-folksonomies}

\author{Mohamed Nader Jelassi and Sadok Ben Yahia and Engelbert Mephu Nguifo}

\maketitle


\begin{abstract}
\emph{Folksonomy} mining is grasping the interest of web $2$.$0$ community since it represents the core data of social resource sharing systems. However, a scrutiny of the related works interested in mining \emph{folksonomies} unveils that the time stamp dimension has not been considered. For example, the wealthy number of works dedicated to mining tri-concepts from \emph{folksonomies} did not take into account time dimension. In this paper, we will consider a \emph{folksonomy} commonly composed of triples $<$users, tags, resources$>$ and we shall consider the time as a new dimension. We motivate our approach by highlighting the battery of potential applications. Then, we present the foundations for mining quadri-concepts, provide a formal definition of the problem and introduce a new efficient algorithm, called \textsc{QuadriCons} for its solution to allow for mining \emph{folksonomies} in time, \emph{i.e.,} \emph{d-folksonomies}. We also introduce a new closure operator that splits the induced search space into equivalence classes whose smallest elements are the quadri-minimal generators. Carried out experiments on large-scale real-world datasets highlight good performances of our algorithm.
\end{abstract}

\IEEEpeerreviewmaketitle

\section{Introduction}\label{sec1}

\emph{Folksonomy} (from folk and taxonomy) is a neologism for a practice of collaborative categorization using freely chosen keywords \cite{Mika2005aInf}. Folksonomies (also called social tagging mechanisms) have been implemented in a number of online knowledge sharing environments since the idea was first adopted by the social bookmarking site \textsc{del.icio.us} in $2004$. The idea of a \emph{folksonomy} is to allow the users to describe a set of shared objects with a set of keywords, \emph{i.e.,} tags, of their own choice. The new data of \textit{folksonomy} systems provides a rich resource for data analysis, information retrieval, and knowledge discovery applications. The rise of \emph{folksonomies}, due to the success of the social resource sharing systems (\emph{e.g.,} \textsc{Flickr}, \textsc{Bibsonomy}, \textsc{Youtube}, etc.) also called Web $2$.$0$, has attracted interest of researchers to deal with the \emph{Folksonomy mining} area. However, due to the huge size of \emph{folksonomies}, many works focus on the extraction of lossless concise representations of interesting patterns, \emph{i.e.,} triadic concepts \cite{Jaschke08} \cite{Ji06cube} \cite{datapeelertkdd}.

Recently, in \cite{pakdd}, the new \textsc{TriCons} algorithm outperforms its competitors thanks to a clever sweep of the search space. Nevertheless, a scrutiny of these related work unveils that the time stamp dimension has not been considered yet.
Time is considered one of the most important factors in detecting emerging subjects. Agrawal and Srikant show in \cite{sequential} the importance of \emph{sequential patterns} which may be useful to discover rules integrating the notion of temporality and sequence of events. In our case, such rules shall be of the form : users which shared the movie "Alcatraz" using the tag \emph{prison} will shared it later with the tag \emph{escape}.

With this paper, we initiate the confluence of threee lines of research, \emph{Formal Concept Analysis}, \emph{Folksonomy mining} and \emph{Mining Sequential Patterns}. \emph{Formal Concept Analysis} (\emph{FCA}) \cite{Ganter99} has been extended since fifteen years ago to deal with three-dimensional data \cite{fritz1995triadic}. However, \emph{Triadic Concept Analysis} (\emph{TCA}) has not garnered much attention for researchers until the coming of \emph{folksonomies} as they represent the core data
 structure of \emph{social networks}. Thus, we give a formal definition of the problem of mining all frequent quadri-concepts (the four-dimensional and sequential version of mining all frequents tri-concepts) and introduce our algorithm \textsc{QuadriCons} for its solution, which is an extension of the \textsc{TriCons} algorithm to the quadratic case. We also introduce a new closure operator that splits the induced search space into equivalence classes whose smallest elements are the quadri-minimal generators (\emph{QGs}); \emph{QGs} are helpful for a clever sweep of the search space \cite{pakdd} \cite{mingen}.

The remainder of the paper is organized as follows. In the next section, we motivate our conceptual and temporal clustering approach for solving the problem of mining all frequent quadri-concepts of a given dataset. We thoroughly study the related work in Section \ref{sec3}. In Section \ref{sec4}, we provide a formal definition of the problem of mining all frequent quadri-concepts. We introduce a new closure operator for the quadratic context as well as the \textsc{QuadriCons} algorithm dedicated to the extraction of all frequent quadri-concepts, in Section \ref{sec5}. In Section \ref{sec6}, carried out experiments about performances of our algorithm in terms of execution time, consumed memory and compacity of the quadri-concepts. Finally, we conclude the paper with a summary and we sketch some avenues for future works in Section \ref{sec7}.


\section{Motivation : Conceptual and Temporal Clustering of Folksonomies}\label{sec2}

The immediate success of social networks, \emph{i.e.,} social resource sharing systems is due to the fact that no specific skills are needed for participating \cite{Jaschke08}. Each individual user is able to share a web page\footnote{http://del.icio.us}, a personal photo\footnote{http://flickr.com}, an artist he likes\footnote{http://last.fm} or a movie he watched\footnote{http://movielens.org} without much effort.

The core data structure of such systems is a \emph{folksonomy}. It consists of three sets $\mathcal{U}$, $\mathcal{T}$, $\mathcal{R}$ of users assigning tags to resources as well as a ternary relation $\mathrm{Y}$ between them. To allow conceptual and temporal clustering from \emph{folksonomies}, an additional dimension, \emph{i.e.,} $\mathcal{D}$, is needed : \textbf{time}.
Indeed, the special feature of \emph{folksonomies} under study is their unceasing evolution \cite{cattuto2007network}. Such systems follow trends and evolve according to the new user's taggings \cite{Hotho2006}. The increasing use of these systems shows that \emph{folksonomy}-based works are then able to offer a better solution in the domain of Web Information Retrieval (WIR) \cite{krause2008comparison} by considering time when dealing with a query or during the user's taggings, \emph{i.e.,} by suggesting the appropriate trendy tags.
Thus, a user which tagged a film or a website with a given tag at a specific date may assign a whole new tag at a different period under completely different circumstances. For example, a user that associate the website \emph{whitehouse.gov} with the tags \emph{Bush} and \emph{Iraq} in $2004$ might assign it the tags \emph{Obama} and \emph{crisis} nowadays. A more real and sadly true example leads users today associating \emph{Islam} with the tag \emph{terrorism} instead of \emph{Quran}; besides, one may see the incessant evolution of the tag \emph{Binladen} in social networks since September $2001$ \cite{Amitay}.

Within the new introduced dimension, \emph{i.e.,} time, our goal is to detect hidden sequential conceptualizations in \emph{folksonomies}. An exemple of such a concept is that users which tagged \emph{"Harry Potter"} will tag \emph{"The Prisoner of Azkaban"} and then tag \emph{"The Order of the Phoenix"}, probably with the same tags.

Our algorithm solves the problem of frequent closed patterns mining for this kind of data. It will return a set of (frequent) quadruples, where each quadruple ($U$, $T$, $R$, $D$) consists of a set $U$ of users, a set $T$ of tags, a set $R$ of resources and a set $D$ of dates. These quadruples, called \emph{(frequent) quadri-concepts}, have the property that each user in $U$ has tagged each resource in $R$ with all tags from $T$ at different dates from $D$, and that none of these  sets can be extended without shrinking one of the other three dimensions. Hence, they represent the four-dimensional and sequential extension of tri-concepts. Moreover, we can add minimum support constraints on each of the four dimensions in order to focus on the largest concepts of the \emph{folksonomy}, \emph{i.e.,} by setting higher values of minimum supports.

In the remainder, we will scrutinize the state-of-the-art propositions aiming to deal with the \emph{folksonomy mining} area.

\section{Related Work}\label{sec3}

In this section, we discuss the different works that deal with \emph{folksonomy} mining.
Due to their triadic form, many researchers \cite{Jaschke08} \cite{Ji06cube} \cite{datapeelertkdd} focus on \emph{folksonomies} in order to extract triadic concepts which are maximal sets of users, tags and resources. Tri-concepts are the first step to a various of applications : ontology building \cite{Mika2005aInf}, association rule derivation \cite{hothorules}, recommendation systems \cite{recommand} to cite but a few. Other papers focus on analysing the structure of folksonomies \cite{Golder06structureCollaborative} or structure the tripartite network of \emph{folksonomies} \cite{cattuto2007network}. Recent works analyse the \emph{folksonomy}'s evolution through time in order to discover the emergent subjects and follow trends \cite{Amitay} \cite{Hotho06trenddetection} \cite{Dubinko}.

Since we are going to mine quadri-concepts from \emph{d-folksonomies}, which mimic the structure of quadratic contexts, we look for works that deal with the four-dimensional data. In \cite{vout2}, inspired by work of Wille \cite{fritz1995triadic} extending \emph{Formal Concept Analysis} to three dimensions, the author created a framework for analyzing $n$-dimensional formal concepts. He generalized the triadic concept analysis to $n$ dimensions for arbitrary $n$, giving rise to Polyadic Concept Analysis. The $n$-adic contexts give rise, in a way analogous to the triadic case, to $n$-adic formal concepts. In \cite{vout2}, the author gives examples of quadratic concepts and their associated quadri-lattice. Despite robust theoretical study, no algorithm has been proposed by Voutsadakis for an efficient extraction of such $n$-adic concepts. Recently, Cerf \emph{et al.} proposed the \textsc{Data-Peeler} algorithm \cite{datapeelertkdd} in order to extract all closed concepts from $n$-ary relations. \textsc{Data-Peeler} enumerates all the $n$-adic formal concepts in a depth first manner using a binary tree enumeration strategy. When setting $n$ to $4$, \textsc{Data-Peeler} is able to extract quadri-concepts.

In the following, we give a formal definition of the problem of mining all frequent quadri-concepts as well as the main notions used through the paper.

\section{The Problem of Mining all Frequent Quadri-Concepts}\label{sec4}

In this section, we formalize the problem of mining all frequents quadri-concepts. We start with an adaptation of the notion of \emph{folksonomy} \cite{Jaschke08} to the quadratic context.
\begin{definition}\label{dfolk} (\textsc{D-Folksonomy})
A \emph{d-folksonomy} is a set of tuples $\mathbb{F}_{d}$ $=$ \textsc{(}$\mathcal{U}$, $\mathcal{T}$, $\mathcal{R}$, $\mathcal{D}$, $\mathrm{Y}$\textsc{)} where $\mathcal{U}$, $\mathcal{T}$, $\mathcal{R}$ and $\mathcal{D}$ are finite sets which elements are called users, tags, resources and dates. $\mathrm{Y}$ $\subseteq$ $\mathcal{U}$ $\times$ $\mathcal{T}$ $\times$ $\mathcal{R}$ $\times$ $\mathcal{D}$ represents a quaternary relation where each $y$ $\subseteq$ $\mathrm{Y}$ can be represented by a quadruple : $y$ = \{($u$, $t$, $r$, $d$) $\mid$ $u$ $\in$ $\mathcal{U}$, $t$ $\in$ $\mathcal{T}$, $r$ $\in$ $\mathcal{R}$, $d$ $\in$ $\mathcal{D}$\} which means that the user $u$ has annotated the resource $r$ using the tag $t$ at the date $d$.
\end{definition}

\begin{example}
Table \ref{quadri} depicts an example of a \textit{d-folksonomy} $\mathbb{F}_{d}$ with  $\mathcal{U}$ = \{$u_{1}$, $u_{2}$, $u_{3}$, $ u_{4}$\}, $\mathcal{T}$ = \{$t_{1}$, $t_{2}$, $t_{3}$\}, $\mathcal{R}$ = \{$r_{1}$, $r_{2}$\} and $\mathcal{D}$ = \{$d_{1}$, $d_{2}$\}. Each cross within the quaternary relation indicates a tagging operation by a user from $\mathcal{U}$, a tag from $\mathcal{T}$ and a resource from $\mathcal{R}$ at a date from $\mathcal{D}$, \textit{i.e.}, a user has tagged a particular resource with a particular tag at a date $d$. For example, the user $u_{1}$ has tagged the resource $r_{1}$ with the tags \textsl{$t_{1}$}, \textsl{$t_{2}$} and \textsl{$t_{3}$} at the date $d_{1}$.

\begin{table}[h]
\begin{center}
{\small
\begin{tabular}{|l|l||l|l|l||l|l|l|}
\hline\hline $\mathbb{F}_{d}$ & $\mathcal{R}$ & \multicolumn{3}{|c|}{\textsl{$r_{1}$}} &
\multicolumn{3}{|c|}{\textsl{$r_{2}$}}  \\
\hline
$\mathcal{D}$  & $\mathcal{U}$/$\mathcal{T}$ & \textsl{$t_{1}$}& \textsl{$t_{2}$}& \textsl{$t_{3}$}& \textsl{$t_{1}$}& \textsl{$t_{2}$}&
\textsl{$t_{3}$}\\

 \hline
      &  \textsl{$u_{1}$} &  $\times$     &       $\times$ &    $\times$      &    $\times$      &  $\times$   &   $\times$ \\
\hline
   \textsl{$d_{1}$}  &   \textsl{$u_{2}$}  &  $\times$     &       $\times$ &      &    $\times$      &  $\times$   &   \\
\hline
      &  \textsl{$u_{3}$} &     &       $\times$ &    $\times$      &        &  $\times$   &   $\times$\\
\hline
      &  \textsl{$u_{4}$}&  $\times$     &       $\times$ &    $\times$      &    $\times$      &  $\times$   &   $\times$\\

\hline
\hline
      &  \textsl{$u_{1}$} &  $\times$     &       $\times$ &        &         &  $\times$   &   \\
\hline
     \textsl{$d_{2}$}   &   \textsl{$u_{2}$} &  $\times$     &       $\times$ &        &         &  $\times$   &   \\
\hline
      &  \textsl{$u_{3}$} &       &       &    $\times$      &      &     &   $\times$\\
\hline
   &  \textsl{$u_{4}$}&       &       &    $\times$      &      &     &   $\times$\\

\hline
\end{tabular}}
\end{center}
\caption{A \textit{d-folksonomy}.}\label{quadri}
\end{table}

\end{example}

The following definition introduces a (frequent) quadri-set.

\begin{definition}\label{qtriset} (\textsc{A (Frequent) quadri-set})
 Let $\mathbb{F}_{d}$ = ($\mathcal{U}$, $\mathcal{T}$, $\mathcal{R}$, $\mathcal{D}$, $\mathrm{Y}$) be a \emph{d-folksonomy}. A quadri-set of $\mathbb{F}_{d}$ is a quadruple ($A$, $B$, $C$, $E$) with $A$ $\subseteq$ $\mathcal{U}$, $B$ $\subseteq$ $\mathcal{T}$, $C$ $\subseteq$ $\mathcal{R}$ and $E$ $\subseteq$
$\mathcal{D}$ such as $A$ $\times$ $B$ $\times$ $C$ $\times$ $E$ $\subseteq$ $\mathrm{Y}$.
\end{definition}

\emph{D-Folksonomies} have four dimensions which are completely symmetric. Thus, we can define minimum support thresholds on each dimension. Hence, the problem of mining frequent quadri-sets is then the following:

\begin{problem}\label{pb1} (\textbf{Mining all frequent quadri-sets})
 Let $\mathbb{F}_{d}$ = ($\mathcal{U}$, $\mathcal{T}$, $\mathcal{R}$, $\mathcal{D}$, $\mathrm{Y}$) be a \emph{d-folksonomy} and let \textit{$minsupp_{u}$}, \textit{$minsupp_{t}$}, \textit{$minsupp_{r}$} and \textit{$minsupp_{d}$} be (absolute) user-defined minimum thresholds. The task of mining all frequent quadri-sets consists in determining all quadri-sets ($A$, $B$, $C$, $E$) of $\mathbb{F}_{d}$ with $|$ $A$ $|$ $\geq$ \textit{$minsupp_{u}$}, $|$ $B$ $|$ $\geq$ \textit{$minsupp_{t}$}, $|$ $C$ $|$ $\geq$ \textit{$minsupp_{r}$} and $|$ $E$ $|$ $\geq$ \textit{$minsupp_{d}$}.
\end{problem}

Our thresholds are antimonotonic constraints : If ($A_{1}$, $B_{1}$, $C_{1}$, $E_{1}$) with $A_{1}$ being maximal for $A_{1}$ $\times$ $B_{1}$ $\times$ $C_{1}$ $\times$ $E_{1}$ $\subseteq$ $\mathrm{Y}$ is not \emph{u-frequent}\footnote{with regard to the dimension $\mathcal{U}$.} then all ($A_{2}$, $B_{2}$, $C_{2}$, $E_{2}$) with $B_{1}$ $\subseteq$ $B_{2}$, $C_{1}$ $\subseteq$ $C_{2}$ and $E_{1}$ $\subseteq$ $E_{2}$ are not \emph{u-frequent} either. The same holds symmetrically for the other three dimensions. In \cite{fritz1995triadic}, the authors demonstrate that above the two-dimensional case, the direct symmetry between monotonicity and antimonotonicity breaks. Thus, they introduced a lemma which results from the triadic Galois connection \cite{Biedermann97a} induced by a triadic context. In the following, we adapt that lemma to our quadratic case.

\begin{lemma}\label{lemme} (See also \cite{vout2}, Proposition $2$)
Let ($A_{1}$, $B_{1}$, $C_{1}$, $E_{1}$) and ($A_{2}$, $B_{2}$, $C_{2}$, $E_{2}$) be quadri-sets with $A_{i}$ being maximal for $A_{i}$ $\times$ $B_{i}$ $\times$ $C_{i}$ $\times$ $E_{i}$ $\subseteq$ $\mathrm{Y}$, for $i$ = $1$,$2$. If $B_{1}$ $\subseteq$ $B_{2}$, $C_{1}$ $\subseteq$ $C_{2}$ and $E_{1}$ $\subseteq$ $E_{2}$ then $A_{2}$ $\subseteq$ $A_{1}$. The same holds symmetrically for the other three dimensions. In the sequel, the inclusion ($A_{1}$, $B_{1}$, $C_{1}$, $E_{1}$) $\subseteq$ ($A_{2}$, $B_{2}$, $C_{2}$, $E_{2}$) holds if and only if $B_{1}$ $\subseteq$ $B_{2}$, $C_{1}$ $\subseteq$ $C_{2}$, $E_{1}$ $\subseteq$ $E_{2}$ and $A_{2}$ $\subseteq$ $A_{1}$.
\end{lemma}


\begin{example}
Let $\mathbb{F}_{d}$ be the \textit{d-folksonomy} of Table \ref{quadri} and let $S_{1}$ = \{\{$u_{3}$, $u_{4}$\}, $t_{3}$, \{$r_{1}$, $r_{2}$\}, \{$d_{1}$, $d_{2}$\}\} and $S_{2}$ = \{\{$u_{1}$, $u_{3}$, $u_{4}$\}, \{$t_{2}$, $t_{3}$\}, \{$r_{1}$, $r_{2}$\}, $d_{1}$\} be two quadri-sets of $\mathbb{F}_{d}$. Then, we have $S_{1}$ $\subseteq$ $S_{2}$ since \{$u_{3}$, $u_{4}$\} $\subseteq$ \{$u_{1}$, $u_{3}$, $u_{4}$\}, $t_{3}$ $\subseteq$ \{$t_{2}$, $t_{3}$\}, \{$r_{1}$, $r_{2}$\} $\subseteq$ \{$r_{1}$, $r_{2}$\} and $d_{1}$ $\subseteq$ \{$d_{1}$, $d_{2}$\}.
\end{example}

As the set of all frequent quadri-sets is highly redundant, we consider a specific condensed representation, \textit{i.e.,} a subset which contains the same information : the set of all frequent quadri-concepts. The latter's definition is given as follows :

\begin{definition}\label{qc} (\textsc{(Frequent) quadratic concept})
A quadratic concept (or a quadri-concept for short) of a \emph{d-folksonomy} $\mathbb{F}_{d}$ = \textsc{(}$\mathcal{U}$, $\mathcal{T}$, $\mathcal{R}$, $\mathcal{D}$, $\mathrm{Y}$\textsc{)} is a quadruple ($U$, $T$, $R$, $D$) with $U$ $\subseteq $ $\mathcal{U}$, $T$ $\subseteq$ $\mathcal{T}$, $R$ $\subseteq$ $\mathcal{R}$ and $D$ $\subseteq$ $\mathcal{D}$ with $U$ $\times$ $T$ $\times$ $R$ $\times$ $D$ $\subseteq$ $\mathrm{Y}$ such that the quadruple ($U$,
 $T$, $R$, $D$) is maximal, \emph{i.e.,} none of these sets can be extended without shrinking one of the other three dimensions. A quadri-concept is said to be frequent whenever it is a frequent quadri-set.
\end{definition}

\begin{problem}\label{pb2} (\textbf{Mining all frequent quadri-concepts})
 Let $\mathbb{F}_{d}$ = ($\mathcal{U}$, $\mathcal{T}$, $\mathcal{R}$, $\mathcal{D}$, $\mathrm{Y}$) be a \emph{d-folksonomy} and let \textit{$minsupp_{u}$}, \textit{$minsupp_{t}$}, \textit{$minsupp_{r}$} and \textit{$minsupp_{d}$} be user-defined minimum thresholds. The task of mining all frequent quadri-concepts consists in determing all quadri-concepts ($U$, $T$, $R$, $D$) of $\mathbb{F}_{d}$ with $|$ $U$ $|$ $\geq$ \textit{$minsupp_{u}$}, $|$ $T$ $|$ $\geq$ \textit{$minsupp_{t}$}, $|$ $R$ $|$ $\geq$ \textit{$minsupp_{r}$} and $|$ $D$ $|$ $\geq$ \textit{$minsupp_{d}$}. The set of all frequent quadri-concepts of $\mathbb{F}_{d}$ is equal to $\mathcal{QC}$ = \{$qc$ $\mid$ $qc$ = ($U$, $T$, $R$, $D$) is a frequent quadri-concept\}.
\end{problem}

\begin{remark}
It is important to note that the extracted representation of quadri-concepts is information lossless. Hence, after solving \emph{Problem} \ref{pb2}, we can easily solve the \emph{Problem} \ref{pb1} by enumerating all quadri-sets ($A$, $B$, $C$, $E$) such as it exists a frequent quadri-concept ($U$, $T$, $R$, $D$) such as $A$ $\subseteq$ $U$, $B$ $\subseteq$ $T$, $C$ $\subseteq$ $R$, $E$ $\subseteq$ $D$ and $|$ $A$ $|$ $\geq$ \textit{$minsupp_{u}$}, $|$ $B$ $|$ $\geq$ \textit{$minsupp_{t}$}, $|$ $C$ $|$ $\geq$ \textit{$minsupp_{r}$} and $|$ $E$ $|$ $\geq$ \textit{$minsupp_{d}$}.
\end{remark}

In the following, we introduce the \textsc{QuadriCons} algorithm for mining all frequent quadri-Concepts before discussing its performances versus the \textsc{Data-Peeler} algorithm for the quadratic case in the section after.


\section{The \textsc{QuadriCons} Algorithm for Mining all Frequent Quadri-Concepts}\label{sec5}
In this section, we introduce new notions that would be of use throughout the \textsc{QuadriCons} algorithm. Hence, we introduce a new closure operator for a \emph{d-folksonomy} which splits the search space into equivalence classes as well as an extension of the notion of minimal generator \cite{pakdd}. Then, we provide an illustrative example of our algorithm.

\subsection{Main notions of \textsc{QuadriCons}}
Before introducing our closure operator for a \emph{d-folksonomy}/quadratic context, we define a general definition of a closure operator for a $n$-adic context.

\begin{definition}\label{opferm-naire}\textsc{(Closure Operator of a $n$-adic context)}
Let $S$ = ($S_{1}$, $S_{2}$, $\ldots$, $S_{n}$) be a $n$-set, with $S_{1}$ being maximal for $S_{1}$ $\times$ $\ldots$ $\times$ $S_{n}$ $\subseteq$ $\mathrm{Y}$, of a $n$-adic context $\mathbb{K}^{n}$ with $n$ dimensions, \emph{i.e.,} $\mathbb{K}^{n}$ = ($\mathcal{D}_{1}$, $\mathcal{D}_{2}$, $\ldots$, $\mathcal{D}_{n}$, $\mathrm{Y}$).
A mapping \textit{h} is defined as follows :

\textit{h}($S$) = \textit{h}($S_{1}$, $S_{2}$, $\ldots$, $S_{n}$) = ($C_{1}$, $C_{2}$, $\ldots$, $C_{n}$) such as :

$C_{1}$ = $S_{1}$

$\wedge$ $C_{2}$ = \{$C_{2}^{i}$ $\in$ $\mathcal{D}_{2}$ $\mid$ ($c_{1}^{i}$, $C_{2}^{i}$, $c_{3}^{i}$, $\ldots$, $c_{n}^{i}$) $\in$
$\mathrm{Y}$ $\forall$ $c_{1}^{i}$ $\in$ $C_{1}$, $\forall$ $c_{3}^{i}$ $\in$ $S_{3}$, $\ldots$, $\forall$ $c_{n}^{i}$ $\in$ $S_{n}$\}

$\vdots$

$\wedge$ $C_{n}$ = \{$C_{n}^{i}$ $\in$ $\mathcal{D}_{n}$ $\mid$ ($c_{1}^{i}$, $c_{2}^{i}$, $\ldots$, $c_{n-1}^{i}$, $C_{n}^{i}$) $\in$
$\mathrm{Y}$ $\forall$ $c_{1}^{i}$ $\in$ $C_{1}$, $\ldots$ $\forall$ $c_{n-1}^{i}$ $\in$ $C_{n-1}$\}
\end{definition}

\begin{proposition}
\label{propclosure}
$h$ is a closure operator.
\end{proposition}

\begin{proof}
To prove that \emph{h} is a closure operator, we have to prove that this closure operator fulfils the three properties of \textbf{extensivity}, \textbf{idempotency} and \textbf{isotony} \cite{closure}.


($1$) \textbf{Extensivity} :
Let $S$ = ($S_{1}$, $S_{2}$, $\ldots$, $S_{n}$) be a $n$-set of  $\mathbb{K}^{n}$ $\Rightarrow$ \textit{h}($S$) = ($C_{1}$, $C_{2}$, $\ldots$, $C_{n}$) such that : $C_{1}$ = $S_{1}$, $C_{2}$ = \{$C_{2}^{i}$ $\in$ $\mathcal{D}_{2}$ $\mid$ ($c_{1}^{i}$, $C_{2}^{i}$, $c_{3}^{i}$, $\ldots$, $c_{n}^{i}$) $\in$
$\mathrm{Y}$ $\forall$ $c_{1}^{i}$ $\in$ $C_{1}$, $\forall$ $c_{3}^{i}$ $\in$ $S_{3}$, $\ldots$, $\forall$ $c_{n}^{i}$ $\in$ $S_{n}$\} $\supseteq$ $S_{2}$ since $C_{1}$ = $S_{1}$ , $\ldots$, $C_{n}$ = \{$C_{n}^{i}$ $\in$ $\mathcal{D}_{n}$ $\mid$ ($c_{1}^{i}$, $c_{2}^{i}$, $\ldots$, $c_{n-1}^{i}$, $C_{n}^{i}$) $\in$ $\mathrm{Y}$ $\forall$ $c_{1}^{i}$ $\in$ $C_{1}$, $\forall$ $c_{2}^{i}$ $\in$ $C_{2}$, $\ldots$ $\forall$ $c_{n-1}^{i}$ $\in$ $C_{n-1}$\} $\supseteq$ $S_{n}$ since $C_{1}$ = $S_{1}$, $C_{2}$ $\supseteq$ $S_{2}$, $\ldots$, $C_{n-1}$ $\supseteq$ $S_{n-1}$.
Then, $C_{1}$ = $S_{1}$ and $S_{i}$ $\subseteq$ $C_{i}$ for $i$ = $2$, $\ldots$ $n$ $\Rightarrow$ $S$ $\subseteq$ $\textit{h}$($S$) (\emph{cf.}, Lemma \ref{lemme})
\\

($2$) \textbf{Idempotency} :
Let $S$ = ($S_{1}$, $S_{2}$, $\ldots$, $S_{n}$) be a $n$-set of $\mathbb{K}^{n}$ $\Rightarrow$ \textit{h}($S$) = ($C_{1}$, $C_{2}$, $\ldots$, $C_{n}$) $\Rightarrow$ \textit{h}($C_{1}$, $C_{2}$, $\ldots$, $C_{n}$) = ($C'_{1}$, $C'_{2}$, $\ldots$, $C'_{n}$) such that : $C'_{1}$ = $C_{1}$, $C'_{2}$ = \{$C_{2}^{i'}$ $\in$ $\mathcal{D}_{2}$ $\mid$ ($c_{1}^{i}$, $C_{2}^{i'}$, $c_{3}^{i}$, $\ldots$, $c_{n}^{i}$) $\in$ $\mathrm{Y}$ $\forall$ $c_{1}^{i}$ $\in$ $C_{1}$, $\forall$ $c_{3}^{i}$ $\in$ $S_{3}$, $\ldots$, $\forall$ $c_{n}^{i}$ $\in$ $S_{n}$\} = $C_{2}$ since $C_{1}$ = $S_{1}$, $\ldots$, $C'_{n}$ = \{$C_{n}^{i'}$ $\in$ $\mathcal{D}_{n}$ $\mid$ ($c_{1}^{i'}$, $c_{2}^{i'}$, $\ldots$, $c_{n-1}^{i'}$, $C_{n}^{i'}$) $\in$ $\mathrm{Y}$ $\forall$ $c_{1}^{i'}$ $\in$ $C'_{1}$, $\forall$ $c_{2}^{i'}$ $\in$ $C'_{2}$, $\ldots$ $\forall$ $c_{n-1}^{i'}$ $\in$ $C'_{n-1}$\} = $C_{n}$ since we have $C'_{1}$ = $C_{1}$, $C'_{2}$ = $C_{2}$, $\ldots$, $C'_{n-1}$ = $C_{n-1}$.
Then, $C'_{i}$ = $C_{i}$ for $i$ = $1$, $\ldots$ $n$  $\Rightarrow$ $\textit{h}(\textit{h}(S))$ = $\textit{h}(S)$
\\

($3$) \textbf{Isotony} :
Let $S$ = ($S_{1}$, $S_{2}$, $\ldots$, $S_{n}$) and $S'$ = ($S'_{1}$, $S'_{2}$, $\ldots$, $S'_{n}$) be two $n$-sets of $\mathbb{K}^{n}$ with $S$ $\subseteq$ $S'$, i.e., $S'_{1}$ $\subseteq$ $S_{1}$ and $S_{i}$ $\subseteq$ $S'_{i}$ for $i$ = $2$, $\ldots$ $n$ (\emph{cf.}, Lemma \ref{lemme}). We have \textit{h}($S$) = ($C_{1}$, $C_{2}$, $\ldots$, $C_{n}$) and \textit{h}($S'$) = ($C'_{1}$, $C'_{2}$, $\ldots$, $C'_{n}$) such that :
\begin{itemize}
  \item $C_{1}$ = $S_{1}$, $C'_{1}$ = $S'_{1}$ and $S'_{1}$ $\subseteq$ $S_{1}$ $\Rightarrow$ $C'_{1}$ $\subseteq$ $C_{1}$

 \item $C_{2}$ = \{$C_{2}^{i}$ $\in$ $\mathcal{D}_{2}$ $\mid$ ($c_{1}^{i}$, $C_{2}^{i}$, $c_{3}^{i}$, $\ldots$, $c_{n}^{i}$) $\in$
$\mathrm{Y}$ $\forall$ $c_{1}^{i}$ $\in$ $C_{1}$, $\forall$ $c_{3}^{i}$ $\in$ $S_{3}$, $\ldots$, $\forall$ $c_{n}^{i}$ $\in$ $S_{n}$\} and $C'_{2}$ = \{$C_{2}^{i'}$ $\in$ $\mathcal{D}_{2}$ $\mid$ ($c_{1}^{i}$, $C_{2}^{i'}$, $c_{3}^{i}$, $\ldots$, $c_{n}^{i}$) $\in$ $\mathrm{Y}$ $\forall$ $c_{1}^{i}$ $\in$ $C_{1}$, $\forall$ $c_{3}^{i}$ $\in$ $S_{3}$, $\ldots$, $\forall$ $c_{n}^{i}$ $\in$ $S_{n}$\} $\Rightarrow$  $C_{2}$ $\subseteq$ $C'_{2}$ since $S_{i}$ $\subseteq$ $S'_{i}$ for $i$ = $3$, $\ldots$ $n$ and $C'_{1}$ $\subseteq$ $C_{1}$. (\emph{cf.}, Lemma \ref{lemme})

  $\vdots$

  \item $C_{n}$ = \{$C_{n}^{i}$ $\in$ $\mathcal{D}_{n}$ $\mid$ ($c_{1}^{i}$, $c_{2}^{i}$, $\ldots$, $c_{n-1}^{i}$, $C_{n}^{i}$) $\in$
$\mathrm{Y}$ $\forall$ $c_{1}^{i}$ $\in$ $C_{1}$, $\forall$ $c_{2}^{i}$ $\in$ $C_{2}$, $\ldots$ $\forall$ $c_{n-1}^{i}$ $\in$ $C_{n-1}$\} and $C'_{n}$ = \{$C_{n}^{i'}$ $\in$ $\mathcal{D}_{n}$ $\mid$ ($c_{1}^{i'}$, $c_{2}^{i'}$, $\ldots$, $c_{n-1}^{i'}$, $C_{n}^{i'}$) $\in$ $\mathrm{Y}$ $\forall$ $c_{1}^{i'}$ $\in$ $C'_{1}$, $\forall$ $c_{2}^{i'}$ $\in$ $C'_{2}$, $\ldots$ $\forall$ $c_{n-1}^{i'}$ $\in$ $C'_{n-1}$\} $\Rightarrow$ $C_{n}$ $\subseteq$ $C'_{n}$ since $C'_{1}$ $\subseteq$ $C_{1}$, $C_{2}$ $\subseteq$ $C'_{2}$, $\ldots$, $C_{n-1}$ $\subseteq$ $C'_{n-1}$. (\emph{cf.}, Lemma \ref{lemme})
\end{itemize}

Then, $C'_{1}$ $\subseteq$ $C_{1}$ and $C_{i}$ $\subseteq$ $C'_{i}$ for $i$ = $2$, $\ldots$ $n$ $\Rightarrow$ $\textit{h}(S)$ $\subseteq$ $\textit{h}(S')$.\\
According to ($1$), ($2$) and ($3$), $\textit{h}$ is a closure operator.
\end{proof}

For $n$ = $4$, we instantiate the closure operator of a quadratic context, \emph{i.e.,} a \emph{d-folksonomy} as follows :

\begin{definition}\label{opferm} \textsc{(Closure operator of a \emph{d-folksonomy})}
Let $S$ = ($A$, $B$, $C$, $E$) be a quadri-set of $\mathbb{F}_{d}$ with $A$ being maximal for $A$ $\times$ $B$ $\times$ $C$ $\times$ $E$ $\subseteq$ $\mathrm{Y}$. The closure operator $h$ of a \emph{d-folksonomy} $\mathbb{F}_{d}$ is defined as follows:

\textit{h}($S$) = \textit{h}($A$, $B$, $C$, $E$) = ($U$, $T$, $R$, $D$)
$|$ $U$ = $A$

$\wedge$ $T$ = \{$t_{i}$ $\in$
$\mathcal{T}$ $\mid$ ($u_{i}$, $t_{i}$, $r_{i}$, $d_{i}$) $\in$
$\mathrm{Y}$ $\forall$ $u_{i}$ $\in$ $U$,
$\forall$ $r_{i}$ $\in$ $C$, $\forall$ $d_{i}$ $\in$ $E$\}

$\wedge$ $R$ = \{$r_{i}$ $\in$
$\mathcal{R}$ $\mid$ ($u_{i}$, $t_{i}$, $r_{i}$, $d_{i}$) $\in$
$\mathrm{Y}$ $\forall$ $u_{i}$ $\in$ $U$,
$\forall$ $t_{i}$ $\in$ $T$, $\forall$ $d_{i}$ $\in$ $E$\}

$\wedge$ $D$ = \{$d_{i}$ $\in$
$\mathcal{D}$ $\mid$ ($u_{i}$, $t_{i}$, $r_{i}$, $d_{i}$) $\in$
$\mathrm{Y}$ $\forall$ $u_{i}$ $\in$ $U$,
$\forall$ $t_{i}$ $\in$ $T$, $\forall$ $r_{i}$ $\in$ $R$\}
\end{definition}

\begin{remark}
Roughly speaking, \textit{h}($S$) computes the largest quadri-set in the \textit{d-folksonomy} $\mathbb{F}_{d}$ which contains maximal sets of tags, resources and dates shared by a group of users. The application of the closure operator $h$ on a quadri-set gives rise to a \emph{quadri-concept} $qc$ = ($U$, $T$, $R$, $D$). In the remainder of the paper, the $U$, $R$, $T$ and $D$ parts are respectively called \textbf{\textit{Extent}}, \textbf{\textit{Intent}}, \textbf{\emph{Modus}} and \textbf{\textit{Variable}}.
\end{remark}


Like the dyadic and triadic case, the closure operator splits the search space into equivalence classes, that we introduce in the following :

\begin{definition}\label{ce} \textsc{\textsc{(}Equivalence class\textsc{)}}
Let $S_{1}$ = ($A_{1}$, $B_{1}$, $C_{1}$, $E_{1}$), $S_{2}$ = ($A_{2}$, $B_{2}$, $C_{2}$, $E_{2}$) be two quadri-sets of $\mathbb{F}_{d}$ and $qc$ be a frequent quadri-concept. $S_{1}$ and $S_{2}$ belong to the same equivalence class represented by the quadri-concept $qc$, \textit{i.e.}, $S_{1}$
$\equiv_{qc}$ $S_{2}$  \textit{iff}  \textit{h}($S_{1}$) = \textit{h}($S_{2}$) = $qc$.
\end{definition}

\begin{figure}[htbp]
\begin{center}
\includegraphics[scale=0.4]{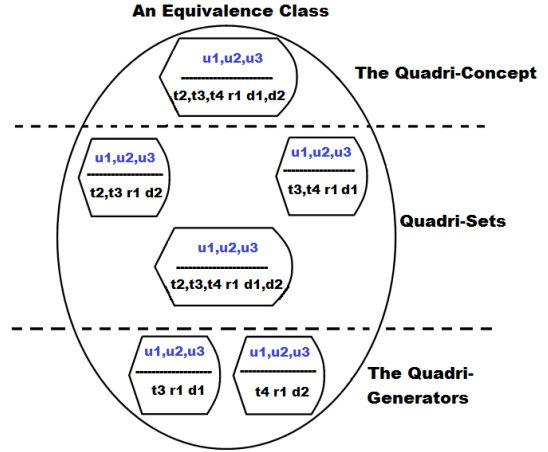}
\end{center}
\caption{Example of an equivalence class extracted from the \emph{d-folksonomy} depicted by Table \ref{quadri}}\label{classeq}
\end{figure}

Minimal Generators (\emph{MGs}) have been shown to play an important role in many theoretical and practical problem settings involving closure systems. Such minimal generators can offer a complementary and simpler way to understand the concept, because they may contain far fewer attributes than closed concepts. Indeed, \emph{MGs} represent the smallest elements within an equivalence class. Complementary to closures, minimal generators provide a way to characterize formal concepts \cite{mingen}. In the following, we introduce an extension of the definition of a \emph{MG} to the \emph{d-folksonomy}.

\begin{definition}\label{min gen} (\textsc{Quadri-Minimal generator})
Let $g$ = ($A$, $B$, $C$, $E$) be a quadri-set of $\mathbb{F}_{d}$ such as $A$ $\subseteq$ $\mathcal{U}$, $B$ $\subseteq$ $\mathcal{T}$, $C$ $\subseteq$ $\mathcal{R}$ and $E$ $\subseteq$ $\mathcal{D}$ and $qc$ $\in$ $\mathcal{QC}$. The quadruple $g$ is a quadri-minimal generator (quadri-generator for short) of $qc$ \textit{iff} \textit{h}($g$) = $qc$ and $\nexists$ $g_{1}$ = ($A_{1}$, $B_{1}$, $C_{1}$, $E_{1}$) such as :
\begin{enumerate}
  \item $A$ = $A_{1}$,
  \item ($B_{1}$ $\subseteq$ $B$ $\wedge$ $C_{1}$ $\subseteq$ $C$ $\wedge$ $E_{1}$ $\subset$ $E$) $\vee$ ($B_{1}$ $\subseteq$ $B$ $\wedge$ $C_{1}$ $\subset$ $C$ $\wedge$ $E_{1}$ $\subseteq$ $E$), and
  \item \textit{h}($g$) = \textit{h}($g_{1}$) = $qc$.
\end{enumerate}

\end{definition}

\begin{example}
Let us consider the \emph{d-folksonomy} $\mathbb{bF}_{d}$ shown in Table \ref{quadri}. Figure \ref{classeq} shows an example of
an equivalence class. For example, we have \textit{h}($g_{1}$=\{\{$u_{1}$, $u_{2}$, $u_{3}$\}, $t_{3}$, $r_{1}$, $d_{1}$\}) = \{\{$u_{1}$, $u_{2}$, $u_{3}$\}, \{$t_{2}$, $t_{3}$, $t_{4}$\}, $r_{1}$, \{$d_{1}$, $d_{2}$\}\} = $qc$ such as $g_{1}$ is a quadri-generator. Thus, $qc$ is the quadri-concept of this equivalence class which is the largest unsubsumed quadri-set and it has two quadri-generators. However, $g_{3}$ = \{\{$u_{1}$, $u_{2}$, $u_{3}$\}, \{$t_{3}$, $t_{4}$\}, $r_{1}$, $d_{1}$\} is not a quadri-generator of $qc$ since it exists $g_{1}$ such as $g_{1}$.\emph{extent}=$g_{3}$.\emph{extent}, $g_{1}$.\emph{intent} = $g_{3}$.\emph{intent} $\wedge$ $g_{1}$.\emph{modus} $\subset$ $g_{3}$.\emph{modus} $\wedge$ $g_{1}$.\emph{variable} = $g_{3}$.\emph{variable}.
\end{example}

Based on those new introduced notions, we propose in the following our new \textsc{QuadriCons} algorithm for a scalable mining of frequent quadri-concepts from a \textit{d-folksonomy}.
\subsection{The \textsc{QuadriCons} Algorithm}

In the following, we introduce a test-and-generate algorithm, called \textsc{QuadriCons}, for mining frequent quadri-concepts from a \textit{d-folksonomy}.
Since quadri-generators are minimal keys of an equivalence class, their detection is largely eased. \textsc{QuadriCons} operates in four steps as follows : the \textsc{FindMinimalGenerators} procedure as a first step for the extraction of quadri-generators. Then, the \textsc{ClosureCompute} procedure is invoked for the three next steps in order to compute respectively the \textit{modus}, \textit{intent} and \textit{variable} parts of quadri-concepts. The pseudo code of the \textsc{QuadriCons} algorithm is sketched by Algorithm \ref{QuadriCons}. \textsc{QuadriCons} takes as input a \textit{d-folksonomy} $\mathbb{F}_{d}$ = ($\mathcal{U}$, $\mathcal{T}$, $\mathcal{R}$, $\mathcal{D}$, $\mathrm{Y}$) as well as four user-defined thresholds (one for each dimension) : \textit{$minsupp_u$}, \textit{$minsupp_t$}, \textit{$minsupp_r$} and \textit{$minsupp_d$}. The output of the \textsc{QuadriCons} algorithm is the set of all frequent quadri-concepts that fulfil these thresholds. \textsc{QuadriCons} works as follows : it starts by invoking the \textsc{FindMinimalGenerators} procedure (Step $1$), which pseudo-code is given by Algorithm \ref{findmg}, in order to extract the quadri-generators stored in the set $\mathcal{MG}$ (Line $3$). For such extraction, \textsc{FindMinimalGenerators} computes for each triple ($t$, $r$, $d$) the set $U_{s}$ representing the maximal set of users sharing both tag $t$ and resource $r$ at the date $d$ (Algorithm \ref{findmg}, Line $3$). If $|U_{s}|$ is frequent \emph{w.r.t} \textit{$minsupp_u$} (Line $4$), a quadri-generator is then created (if it does not already exist) with the appropriate fields (Line $5$). Algorithm \ref{findmg} invokes the \textbf{\textsc{AddQuadri}} function which adds the quadri-generator $g$ to the set $\mathcal{MG}$ (Line $7$).

\begin{algorithm}[htbp]\label{monalgo}
{

\KwData{\begin{enumerate}   \item $\mathbb{F}_{d}$ ($\mathcal{U}$, $\mathcal{T}$, $\mathcal{R}$, $\mathcal{D}$, $\mathrm{Y}${)} : A \textit{d-folksonomy}.
        \item  \textit{$minsupp_u$}, \textit{$minsupp_t$}, \textit{$minsupp_r$}, \textit{$minsupp_d$} : User-defined thresholds.
\end{enumerate}}

\KwResult{ $\mathcal{QC}$ : \{Frequent quadri-concepts\}.}

    \Begin {

    /*\emph{Step $1$ : The extraction of quadri-generators}*/\\
     \textsc{FindMinimalGenerators($\mathbb{F}_{d}$, $\mathcal{MG}$, \textit{$minsupp_u$});}

/*\emph{Step $2$ : The computation of the modus part}*/\\
\ForEach {quadri-gen $g$ $\in$  $\mathcal{MG}$ } {
\textsc{ClosureCompute}($\mathcal{MG}$, \textit{$minsupp_u$}, \textit{$minsupp_t$}, \textit{$minsupp_r$}, $g$, $\mathcal{QS}$, $1$);}

\textsc{PruneInfrequentSets}($\mathcal{QS}$,\textit{$minsupp_t$});


/*\emph{Step $3$ : The computation of the intent part}*/\\
\ForEach {quadri-set $s$ $\in$  $\mathcal{QS}$ } {
\textsc{ClosureCompute}( $\mathcal{QS}$, \textit{$minsupp_u$}, \textit{$minsupp_t$}, \textit{$minsupp_r$}, $s$, $\mathcal{QS}$, $2$);}

\textsc{PruneInfrequentSets}($\mathcal{QS}$,\textit{$minsupp_r$});


/*\emph{Step $4$ : The computation of the variable part}*/\\
\ForEach { quadri-set $s$ $\in$ $\mathcal{QS}$ } {
\textsc{ClosureCompute}( $\mathcal{QS}$, \textit{$minsupp_u$}, \textit{$minsupp_t$}, \textit{$minsupp_r$}, $s$, $\mathcal{QC}$, $3$);}

\textsc{PruneInfrequentSets}($\mathcal{QC}$,\textit{$minsupp_d$});

} \Return $\mathcal{QC}$ ; }

  \caption{\textsc{QuadriCons}}
  \label{QuadriCons}
\end{algorithm}


\begin{algorithm}[htbp]\label{gen}
{

\KwData { \begin{enumerate}
\item $\mathcal{MG}$ : The set of frequent quadri-generators.
\item $\mathbb{F}_{d}$ {(}$\mathcal{U}$, $\mathcal{T}$, $\mathcal{R}$, $\mathcal{D}$, $\mathrm{Y}${)} : A \emph{d-folksonomy}.
\item \textit{$minsupp_u$} : User-defined threshold of user's support.
\end{enumerate}}

\KwResult{ $\mathcal{MG}$ : \{The set of frequent quadri-generators\}.}

    \Begin {
    \ForEach{triple ($t$, $r$, $d$) of $\mathbb{F}_{d}$}{
    $U_{s}$= \{$u_{i}$ $\in$ $\mathcal{U}$ $\mid$ ($u_{i}$, $t$, $r$, $d$) $\in$   $\mathrm{Y}$\} ;\\

      \If {$|$ $U_{s}$ $|$ $\geq$ \textit{$minsupp_u$}} {$g.extent$ = $U_{s}$; $g.intent$ = $r$; $g.modus$ = $t$; $g.variable$ = $d$\\
          \If{$g$ $\not \in$ $\mathcal{MG}$}{ \textsc{AddQuadri}($\mathcal{MG}$, $g$)}}

          }

}
  \Return $\mathcal{MG}$ ;}

  \caption{\textsc{FindMinimalGenerators }}
  \label{findmg}
\end{algorithm}


\begin{algorithm}[!htbp]\label{findmix}
{

\KwData { \begin{enumerate}
\item $\mathcal{S}_{IN}$ : The input set.
\item \textit{$min_u$}, \textit{$min_t$}, \textit{$min_r$} : User-defined thresholds.
\item $q$ : A quadri-generator/quadri-set.
\item $\mathcal{S}_{OUT}$ : The output set.
\item \textit{i} : an indicator.
\end{enumerate}}

\KwResult{ $\mathcal{S}_{OUT}$ : The output set.}

    \Begin {

\ForEach { quadri-set $q'$ $\in$ $\mathcal{S}_{IN}$} {

\If{\emph{i}=$1$ \textbf{and} $q.intent$ = $q'.intent$ \textbf{and} $q.extent$ $\subseteq$ $q'.extent$}
{
$s.intent$ = $q.intent$;$s.extent$ = $q.extent$;$s.variable$ = $q.variable$;$s.modus$ = $q.modus$ $\cup$ $q'.modus$; \textsc{AddQuadri}($\mathcal{S}_{OUT}$, $s$);
}

\ElseIf{\emph{i}=$1$ \textbf{and} $q.intent$ = $q'.intent$ \textbf{and} $q$ and $q'$ \emph{incomparable}}
{
$g.extent$ = $q.extent$ $\cap$ $q'.extent$; $g.modus$ = $q.modus$ $\cup$ $q'.modus$; $g.intent$ = $q.intent$; $g.variable$ = $q.variable$;\\  \textbf{If} $g$ \emph{u-frequent} \textbf{then} \textsc{AddQuadri}($\mathcal{MG}$, $g$);
}

\ElseIf{\emph{i}=$2$ \textbf{and} $q.extent$ $\subseteq$ $q'.extent$ \textbf{and} $q.modus$ $\subseteq$ $q'.modus$ \textbf{and} $q.intent$ $\neq$ $q'.intent$ }
{
$qs.extent$ = $q.extent$;
$qs.modus$ = $q.modus$;
$qs.variable$ = $q.variable$;
$qs.intent$ = $q.intent$ $\cup$ $q'.intent$;\\
\textsc{AddQuadri}($\mathcal{S}_{OUT}$, $qs$);
}

\ElseIf{\emph{i}=$2$ \textbf{and} $q$ and $q'$ \emph{incomparable}}
{
$s.extent$ = $q.extent$ $\cap$ $q'.extent$; $s.modus$ = $q.modus$ $\cap$ $q'.modus$; $s.variable$ = $q.variable$; $s.intent$ = $q.intent$ $\cup$ $q'.intent$;

\textbf{If} $s$ is \emph{u-frequent} and \emph{t-frequent} \textbf{then} \textsc{AddQuadri}($\mathcal{S}_{OUT}$, $s$);
}

\ElseIf{\emph{i}=$3$ \textbf{and} $q.extent$ $\subseteq$ $q'.extent$ \textbf{and} $q.modus$ $\subseteq$ $q'.modus$ \textbf{and} $q.intent$ $\subseteq$ $q'.intent$ \textbf{and} $q.variable$ $\neq$ $q'.variable$ }
{
$qc.extent$ = $q.extent$;
$qc.modus$ = $q.modus$;
$qc.intent$ = $q.intent$;
$qc.variable$ = $q.variable$ $\cup$ $q'.variable$;\\
\textsc{AddQuadri}($\mathcal{S}_{OUT}$, $qc$);
}

\ElseIf{\emph{i}=$3$ \textbf{and} $q$ and $q'$ \emph{incomparable}}
{
$s.extent$ = $q.extent$ $\cap$ $q'.extent$; $s.modus$ = $q.modus$ $\cap$ $q'.modus$; $s.intent$ = $q.intent$ $\cap$ $q'.intent$;  $s.variable$ = $q.variable$ $\cup$ $q'.variable$;

\textbf{If} $s$ is \emph{u-frequent}, \emph{t-frequent} and \emph{r-frequent} \textbf{then} \textsc{AddQuadri}($\mathcal{S}_{OUT}$, $s$);
}

}

}  \Return $\mathcal{S}_{OUT}$ ; }

  \caption{\textsc{ClosureCompute}}
  \label{findrs}
\end{algorithm}


Hereafter, \textsc{QuadriCons} invokes the \textsc{ClosureCompute} procedure (Step $2$) for each quadri-generator of $\mathcal{MG}$ (Lines $5$-$7$), which pseudo-code is given by Algorithm \ref{findmix} : the aim is to compute the \textit{modus} part of each quadri-concept. At this step, the two first cases of Algorithm \ref{findmix} (Lines $3$ and $6$) have to be considered \emph{w.r.t} the \textit{extent} of each quadri-generator. The \textsc{ClosureCompute} procedure returns the set $\mathcal{QS}$ formed by quadri-sets. The indicator \textit{flag} (equal to $1$ here) marked by \textsc{QuadriCons} shows if the quadri-set considered by the \textsc{ClosureCompute} procedure is a quadri-generator. In the third step, \textsc{QuadriCons} invokes a second time the \textsc{ClosureCompute} procedure for each quadri-set of $\mathcal{QS}$ (Lines $9$-$11$), in order to compute the \emph{intent} part. \textsc{ClosureCompute} focuses on quadri-sets of $\mathcal{QS}$ having different \textit{intent} parts (Algorithm \ref{findmix}, Line $10$). The fourth and final step of \textsc{QuadriCons} invokes a last time the \textsc{ClosureCompute} procedure with an indicator equal to $3$. This will allow to focus on quadri-sets having different \textit{variable} parts (Algorithm \ref{findmix}, Line $18$) before generating quadri-concepts. \textsc{QuadriCons} comes to an end after this step and returns the set of the frequent quadri-concepts which fulfils the four thresholds \emph{$minsupp_u$}, \emph{$minsupp_t$}, \emph{$minsupp_r$} and \emph{$minsupp_d$}. The \textsc{QuadriCons} algorithm invokes the \textbf{\textsc{PruneInfrequentSets}} function (Lines $8$, $13$ and $18$) in order to prune infrequent quadri-sets/concepts, \emph{i.e.}, whose the \emph{modus/intent/variable} cardinality does not fulfil the aforementioned thresholds.

\subsection{Structural properties of \textsc{QuadriCons}}

\begin{proposition}\label{correct}
The \textsc{QuadriCons} algorithm is correct and complete. It retrieves accurately all the frequent quadri-concepts.
\end{proposition}

\begin{proof}
The \textsc{FindMinimalGenerators} procedure allows to extract all quadri-generators from the \emph{d-folksonomy} $\mathcal{F}_{d}$ since all the context's triples are enumerated in order to group maximal users \emph{w.r.t} each triple ($t$,$r$,$d$) (Algorithm $2$, Lines $2$-$10$). This allows to extract accurately all the quadri-generators. From quadri-generators already extracted, \textsc{QuadriCons} calls the \textsc{ClosureCompute} procedure three times in order to compute, respectively, the \emph{modus}, \emph{intent} and \emph{variable} parts of each quadri-generator. At each call, \emph{i.e.,} $i$ = $1$, $2$, $3$, for each couple of candidates $q$ and $q'$, two cases have to be considered :

\begin{enumerate}
\item  (Algorithm $3$, lines $3$, $10$, $18$) $q$ and $q'$ are comparable. Hence a quadri-set (quadri-concept when $i$ = $3$) is created from the union of different parts of both candidates.

\item  (Algorithm $3$, lines $6$, $14$, $22$) $q$ and $q'$ are incomparable. Hence, a new quadri-set (quadri-generator when $i$ = $1$) is created matching the different parts of $q$ and $q'$.

\end{enumerate}

Thus, all cases of comparison between candidates are enumerated. Finally, the \textsc{PruneInfrequentSets} procedure prune infrequent quadri-concepts \emph{w.r.t} minimum thresholds (Algorithm $1$, lines $8$, $13$ and $18$). We conclude that \textsc{QuadriCons} faithfully extracts all frequent quadri-concepts. So, it is correct.
\end{proof}

\begin{proposition}
\label{termi}
The \textsc{QuadriCons} algorithm terminates.
\end{proposition}

\begin{proof}
The number of quadri-generators generated by \textsc{QuadriCons} is finite. Indeed, the number of \emph{QGs} candidate generated from a context ($\mathcal{U}$, $\mathcal{T}$, $\mathcal{R}$, $\mathcal{D}$) is at most $|\mathcal{T}|\times|\mathcal{R}|\times|\mathcal{D}|$. Since the set $\mathcal{MG}$ of quadri-generators is finite, the three loops of Algorithm $1$ running this set are thus finite. Moreover, the total number of quadri-concepts generated by \textsc{QuadriCons} is equal to $2^{|\mathcal{T}|+|\mathcal{R}|+|\mathcal{D}|}$ Therefore, the algorithm \textsc{QuadriCons} terminates.
\end{proof}

\medskip
\textbf{Theoretical Complexity issues:}
As in the triadic case \cite{Jaschke08}, the number of (frequent) quadri-concepts may grow exponentially in the worst case. Hence, the theoretical complexity of our algorithm is around $\mathcal{O}$($2^{n}$) with $n$ =  $|\mathcal{T}|+|\mathcal{R}|+|\mathcal{D}|$. Nevertheless, and as it will be shown in the section dedicated to experimental results, from a practical point of view, the actual performances are far from being exponential and \textsc{QuadriCons} flags out the desired scalability feature.  Therefore we focus on empirical evaluations on large-scale real-world datasets.

\subsection{Illustrative example}

Consider the \emph{d-folksonomy} depicted by Table \ref{quadri}, with \textit{$minsupp_u$} = $2$, \textit{$minsupp_t$} = $2$, \textit{$minsupp_r$} = $1$ and \textit{$minsupp_d$} = $1$. Figure \ref{exec} sketches the execution trace of \textsc{QuadriCons} above this context. As described above, \textsc{QuadriCons} operates in four steps :

\begin{enumerate}
  \item \emph{(Step $1$)} The first step of \textsc{QuadriCons} involves the extraction of quadri-generators (\emph{QGs}) from the context (Algorithm $1$, Line $3$). \emph{QGs} are maximal sets of users following a triple of tag, resource and date. Thus, eleven \emph{QGs} (among twelve) fulfill the minimum threshold \textit{$minsupp_u$} (\emph{cf.,} Figure \ref{exec}, Step $1$).

   \item  \emph{(Step $2$)} Next, \textsc{QuadriCons} invokes the \textsc{ClosureCompute} procedure a first time on the quadri-generators allowing the computation of the \textit{modus} part (the set of tags) of such candidates (Algorithm $1$, Lines $5$-$8$). For example, since the \emph{extent} part (the set of users) of \{\{$u_{1}$, $u_{2}$, $u_{4}$\}, $t_{1}$, $r_{1}$, $d_{1}$\} is included into that of \{\{$u_{1}$, $u_{2}$, $u_{3}$, $u_{4}$\}, $t_{2}$, $r_{1}$, $d_{1}$\}, the \emph{modus} part of the first \emph{QG} will be equal to \{$t_{1}$, $t_{2}$\}. In addition, new \emph{QGs} can be created from intersection of the first ones (Algorithm \ref{findrs}, Lines $6$-$9$) : it is the case of the two \emph{QGs} \emph{(a)} and \emph{(b)} (\emph{cf.}, Figure \ref{exec}, Step $2$). Finally, candidates that not fulfill the minimum threshold \textit{$minsupp_t$} are pruned (\emph{cf.}, the three last ones).

  \item  \emph{(Step $3$)} Then, \textsc{QuadriCons} proceeds to the computation of the \textit{intent} part (the set of resources) of each candidate within a second call to the \textsc{ClosureCompute} procedure (Algorithm $1$, Lines $10$-$13$). For example, the candidate \{\{$u_{1}$, $u_{2}$, $u_{4}$\}, \{$t_{1}$, $t_{2}$\}, $r_{1}$, $d_{1}$\} has an \emph{extent}, \emph{modus} and \emph{variable} included or equal into those of the candidate \{\{$u_{1}$, $u_{2}$, $u_{4}$\}, \{$t_{1}$, $t_{2}$\}, $r_{2}$, $d_{1}$\}. Then, its \emph{intent} will be equal to \{$r_{1}$, $r_{2}$\}. At this step, four candidates fulfill the minimum threshold \textit{$minsupp_r$} (\emph{cf.}, Figure \ref{exec}, Step $3$). By merging comparable candidates, this step allow reducing at the same time their number.

  \item  \emph{(Step $4$)} Via a last call to the \textsc{ClosureCompute} procedure, \textsc{QuadriCons} computes the \textit{variable} part (the set of dates) of each candidate while pruning infrequent ones (Algorithm $1$, Lines $15$-$18$). For example, since the candidate \{\{$u_{1}$, $u_{2}$\}, \{$t_{1}$, $t_{2}$\}, $r_{1}$, $d_{2}$\} has an \emph{extent}, \emph{modus} and \emph{intent} included into those of \{\{$u_{1}$, $u_{2}$, $u_{4}$\}, \{$t_{1}$, $t_{2}$\}, \{$r_{1}$, $r_{2}$\}, $d_{1}$\}\footnote{Concretely, it means that the users $u_{1}$ and $u_{2}$ who shared the resource $r_{1}$ with the tags $t_{1}$ and $t_{2}$ at the date $d_{2}$ also shared it at the date $d_{1}$.}, its \emph{variable} will be equal to \{$d_{1}$, $d_{2}$\} (\emph{cf.}, Figure \ref{exec}, Step $4$).
\end{enumerate}

 After the Step $4$, \textsc{QuadriCons} terminates. The four frequent quadri-concepts given as output are :
 \begin{enumerate}
  \item \{\{$u_{1}$, $u_{2}$, $u_{4}$\}, \{$t_{1}$, $t_{2}$\}, \{$r_{1}$, $r_{2}$\}, $d_{1}$\}
  \item \{\{$u_{1}$, $u_{3}$, $u_{4}$\}, \{$t_{2}$, $t_{3}$\}, \{$r_{1}$, $r_{2}$\}, $d_{1}$\}
  \item \{\{$u_{1}$, $u_{4}$\}, \{$t_{1}$, $t_{2}$, $t_{3}$\}, \{$r_{1}$, $r_{2}$\}, $d_{1}$\}
  \item \{\{$u_{1}$, $u_{2}$\}, \{$t_{1}$, $t_{2}$\}, $r_{1}$, \{$d_{1}$, $d_{2}$\}\}
\end{enumerate}

\begin{figure}[htbp]
\begin{center}
\includegraphics[scale=0.6]{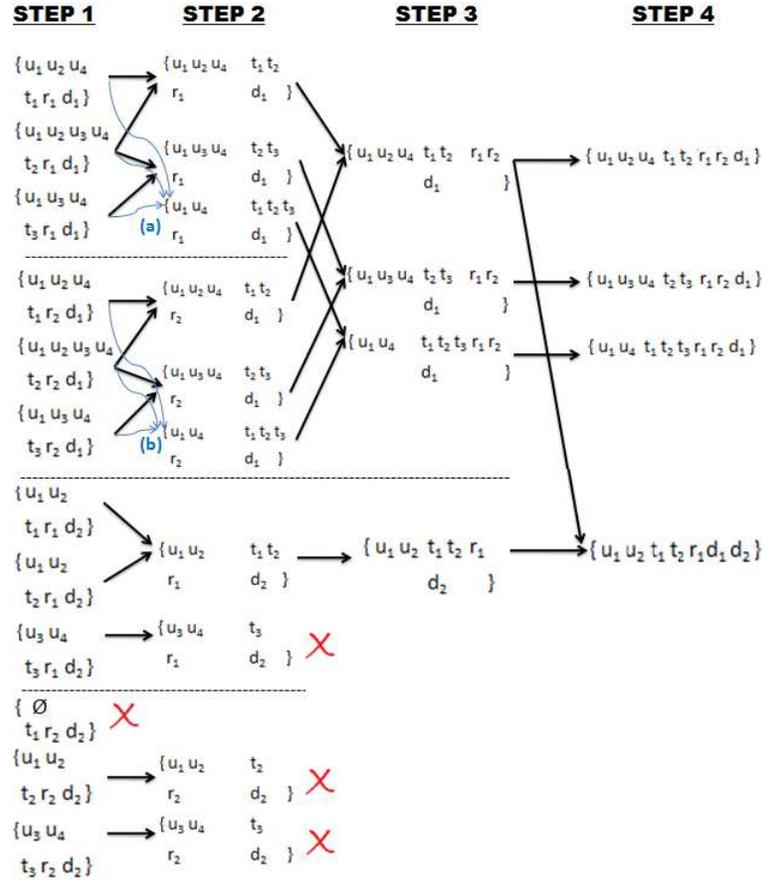}
\end{center}
\caption{Execution trace of \textsc{QuadriCons} above the \emph{d-folksonomy} depicted by Table \ref{quadri}}\label{exec}
\end{figure}


\section{Evaluation and Discussion}\label{sec6}
In this section, we show through extensive carried out experiments, the assessment of the \textsc{QuadriCons} performances \emph{vs.} the state-of-the-art \textsc{Data-Peeler} algorithm in terms of execution time\footnote{All implemented algorithms are in \emph{C++} (compiled with \emph{GCC} $4$.$1$.$2$) and we used an Intel$®$ Core$™$ $i7$ CPU system with $4$ GB RAM. Tests were carried out on the Linux operating system \textsc{Ubuntu} $10$.$10$.$1$.}. We also put the focus on the differences between the consumed memory of both algorithms. Finally, we compare the number of frequent quadri-concepts versus the number of frequent quadri-sets in order to assess the compacity of the extracted representation. We have applied our experiments on two real-world datasets described in the following. Both datasets \cite{movie} are freely downloadable\footnote{\emph{http://movielens.org}} and statistics about these snapshots are summarized into Table \ref{Statsdel}.

\begin{itemize}
\item \textsc{MovieLens} (\emph{http://movielens.org}) is a movie recommendation website. Users are asked to annotate movies they like and dislike. Quadruples are sets of users sharing movies using tags at different dates.
\item \textsc{Last.fm} (\emph{http://last.fm}) is a music website, founded in $2002$. It has claimed $30$ million active users in March $2009$. Quadruples are sets of users annotating artists through tags at different dates.
\end{itemize}

\begin{table}[htbp]
\begin{center}
{\small
\begin{tabular}{|r|r|r|}
\hline

           &   Dataset $1$ & Dataset $2$\\
\hline
       &          (\textsc{MovieLens}) & (\textsc{Last.fm})\\
\hline
  \textbf{Type} &     Dense    & Sparse\\
\hline
$\#$ \textbf{Quadruples}&        $95580$ & $186479$\\
\hline
$\#$ \textbf{Users} &          $4010$ & $1892$\\
\hline
     $\#$ \textbf{Tags} &          $15227$ & $9749$ \\
\hline
 $\#$ \textbf{Resources} &      $11272$ (movies)   & $12523$ (artists)\\
\hline
 $\#$ \textbf{Dates (timestamps)}&      $81601$    & $3549$ \\
\hline
  \textbf{Periods} &      $12/01/2005$ -     & $10/01/2007$ -  \\
 &      $20/12/2008$    & $07/08/2011$ \\
\hline
\end{tabular}}
\end{center}
\caption{Characteristics of the considered snapshots.\label{Statsdel}}
\end{table}

\begin{table}[htbp]
\begin{center}
{\small
\begin{tabular}{|l|l|l|l|l|}
  \hline
\textbf{Datasets} & \textbf{Dates} & \textbf{Users} & \textbf{Tags} & \textbf{Resources}\\
  \hline
   & $03/12/05$ &  &  \emph{kids} & \emph{Harry Potter} \\

   & & \emph{krycek} & \emph{fantasy} & \emph{The Prisoner} \\
  \textsc{Movie} &  $16/07/06$ &  &  \emph{darkness} & \emph{of Azkaban} \\
  \textsc{Lens} & & \emph{maria} & \emph{magic} & \emph{The Order of} \\
     & $21/02/08$   & &  & \emph{the Phoenix} \\

   \hline
   \hline

   & $07/05/10$ & \emph{csmdavis} & \emph{pop} & \emph{Britney Spears} \\
  \textsc{Last.fm} &   & \emph{franny} & \emph{concert} & \emph{Madonna} \\
   & $02/06/11$ & \emph{rossanna} & \emph{dance} & \\

   \hline
\end{tabular}}
\end{center}
\caption{Examples of frequent quadri-concepts of \textsc{MovieLens} and \textsc{Last.fm}.\label{examplesdeqc}}
\end{table}
\vspace{-0.7cm}
\subsection{Examples of quadri-concepts}
Table \ref{examplesdeqc} shows two examples of frequent quadri-concepts extracted from the \textsc{MovieLens} and \textsc{Last.fm} datasets. The first one depicts that the users \emph{krycek} and \emph{maria} used the tags \emph{kids}, \emph{fantasy}, \emph{darkness} and \emph{magic} to annotate the movie \emph{Harry Potter} and its sequels successively in $03/12/2005$, in $16/07/2006$ and then in $21/02/2008$. Such concept may be exploited further for recommending tags for that movie or analyze the evolution of tags associated to \emph{"Harry Potter"}. The second quadri-concept shows that the users \emph{csmdavis}, \emph{franny} and \emph{rossanna} shared the tags \emph{pop}, \emph{concert} and \emph{dance} to describe the artists \emph{Britney Spears} and \emph{Madonna} in $07/05/10$ and then in $02/06/11$. We can use such quadri-concept to recommend the users \emph{franny} and \emph{rossanna} to the first one, \emph{i.e.,} \emph{csmdavis} as they share the same interest for both artists using the same tags. It will be also useful to study the evolution of the artist's \emph{fans} and the vocabulary they used to annotate them through time.



In the following, in order to assess the performances of \textsc{QuadriCons} \emph{vs.} \textsc{Data-Peeler} while extracting quadri-concepts, we ran both algorithms on both datasets and we vary the values of minimum thresholds as depicted by Tables \ref{perfmovie} and \ref{perflast}.

\subsection{Execution Time}
Tables \ref{perfmovie} and \ref{perflast} show the different runtimes of the \textsc{QuadriCons} algorithm \emph{vs.} those of \textsc{Data-Peeler} for the different values of quadruples, which grows from $20000$ to $95580$ for the \textsc{MovieLens} dataset and from $40000$ to $186479$ for the \textsc{Last.fm} dataset, and for different values of minimum thresholds. We can observe that for both datasets and for all values of the number of quadruples,  \textsc{Data-Peeler} algorithm is far away from \textsc{QuadriCons} in terms of execution time. \textsc{QuadriCons} ran until $332$ times faster than \textsc{Data-Peeler} on \textsc{Last.fm} and until $124$ times on \textsc{MovieLens}. Indeed, the poor performance flagged out by \textsc{Data-Peeler}, is explained by the strategy adopted by this later which starts by storing the entire dataset into a \emph{binary tree structure}, which should facilitate its run and then the extraction of quadri-concepts. However, such structure is absolutely not adequate to support a so highly sized data, which is the case of the real-world large-scale datasets considered in our evaluation. Contrariwise, The main thrust of the \textsc{QuadriCons} algorithm stands in the localisation of the quadri-generators (\emph{QGs}), that stand at the "antipodes" of the closures within their respective equivalence classes. Then, in an effort to improve the existing work, our strategy to locate these \emph{QGs} have the advantage of making the extraction of quadri-concepts faster than its competitor. This is even more significant in the case of our real-world datasets where the number of data reaches thousands.

\subsection{Consumed Memory}
Tables \ref{perfmovie} and \ref{perflast} show the memory consumed by both algorithms on both datasets for the different values of quadruples. We observe that \textsc{QuadriCons} consumes memory far below its competitor : less than $40000$ KB and $20000$ KB on both datasets versus millions of KB for \textsc{Data Peeler}. Such difference is explained by the fact that \textsc{QuadriCons}, unlike \textsc{Data Peeler}, does not store the dataset in memory before proceeding the extraction of quadri-concepts. Furthermore, \textsc{QuadriCons} generates fewer candidates thanks to the clever detection of quadri-generators that reduce the search space significantly. For example, to extract the $167$ quadri-concepts from \textsc{Last.fm} when \textit{$minsupp_{u}$} = $3$, \textit{$minsupp_{t}$} = $2$, \textit{$minsupp_{r}$}  = $1$ and \textit{$minsupp_{d}$} = $1$, \textsc{QuadriCons} requires only $1754$ KB in memory while detecting the $939$ quadri-generators of the dataset. However, despite the few number of extracted quadri-concepts, \textsc{Data Peeler} requires $788021$ KB in memory to store the entire dataset before generating candidates. Hence, detecting quadri-generators before extracting quadri-concepts allows \textsc{QuadriCons} consuming until $54$ and $115$ times less memory than \textsc{Data Peeler} on respectively \textsc{MovieLens} and \textsc{Last.fm} datasets.

\begin{table}[h]
\begin{center}
\small
\begin{tabular}{|r||r|r|r|r|}
\hline
 &  {\bf \textsc{Quadri}} & Consumed  & {\bf \textsc{Data}}  & Consumed\\
 $|$ $\mathrm{Y}$ $|$ & {\bf \textsc{Cons}} & Memory & {\bf \textsc{Peeler}} & Memory\\
    & (sec) & (kilobytes) & (sec) & (kilobytes)\\

\hline

\multicolumn{5}{|c|}{\textit{$minsupp_{u}$} = $3$, \textit{$minsupp_{t}$} = $2$,}  \\
\multicolumn{5}{|c|}{\textit{$minsupp_{r}$}  = $1$, \textit{$minsupp_{d}$} = $1$} \\
\hline
\hline
                  $25000$ &      $\textbf{0. 86}$ & $\textbf{542}$    & $43. 10$ & $209843$\\

                 $50000$ &       $\textbf{2. 05}$ & $\textbf{1361}$    & $110. 72$ & $378907$\\

       $70000$ &      $\textbf{3. 08}$ & $\textbf{1760}$    & $198. 33$ & $509541$\\

              $95580$ &      $\textbf{4. 61}$ & $\textbf{2087}$    & $288. 00$ & $654761$\\
\hline
\hline

\multicolumn{5}{|c|}{\textit{$minsupp_{u}$} = $2$, \textit{$minsupp_{t}$} = $2$,}  \\
\multicolumn{5}{|c|}{\textit{$minsupp_{r}$}  = $2$, \textit{$minsupp_{d}$} = $1$} \\
\hline
\hline
                 $25000$ &      $\textbf{0. 36}$ & $\textbf{198}$    & $39. 98$ & $399672$\\

                 $50000$ &       $\textbf{0. 97}$ & $\textbf{431}$    & $107. 71$ & $508943$\\

       $70000$ &      $\textbf{1 .96}$ & $\textbf{567}$    & $227. 65$ & $667006$\\

              $95580$ &      $\textbf{3. 79}$ & $\textbf{1182}$    & $472.87$ & $842551$\\
\hline
\hline

\multicolumn{5}{|c|}{\textit{$minsupp_{u}$} = $2$, \textit{$minsupp_{t}$} = $2$,}  \\
\multicolumn{5}{|c|}{\textit{$minsupp_{r}$}  = $1$, \textit{$minsupp_{d}$} = $1$} \\
\hline
\hline
                  $25000$ &      $\textbf{5.76}$ & $\textbf{2491}$    & $421. 44$ & $769822$\\

                 $50000$ &       $\textbf{15.92}$ & $\textbf{5246}$    & $1269. 70$ & $976200$\\

       $70000$ &      $\textbf{29.22}$ & $\textbf{9845}$    & $2037. 73$ & $1153401$\\

              $95580$ &      $\textbf{48.92}$ & $\textbf{16556}$    & $3478.98$ & $1446242$\\
\hline
\hline

\multicolumn{5}{|c|}{\textit{$minsupp_{u}$} = $2$, \textit{$minsupp_{t}$} = $1$,}  \\
\multicolumn{5}{|c|}{\textit{$minsupp_{r}$}  = $1$, \textit{$minsupp_{d}$} = $1$} \\
\hline
\hline
                 $25000$ &      $\textbf{97. 56}$ & $\textbf{10982}$   & $1022. 12$ & $1272988$\\

                 $50000$ &       $\textbf{188. 61}$ & $\textbf{14671}$    & $1987. 06$ & $1561992$\\

       $70000$ &      $\textbf{263. 63}$ & $\textbf{19548}$    & $2876. 02$ & $1751258$\\

              $95580$ &      $\textbf{528. 58}$ & $\textbf{38762}$    & $5965. 94$ & $2098452$\\
\hline
\hline

\end{tabular}
\end{center}
\caption{Performances of \textsc{QuadriCons} \emph{vs}. \textsc{Data-Peeler} above the \textsc{MovieLens} dataset.}\label{perfmovie}
\end{table}

\begin{table}[h]
\begin{center}
\small
\begin{tabular}{|r||r|r|r|r|}
\hline
 &  {\bf \textsc{Quadri}} & Consumed  & {\bf \textsc{Data}}  & Consumed\\
  $|$ $\mathrm{Y}$ $|$ & {\bf \textsc{Cons}} & Memory & {\bf \textsc{Peeler}} & Memory\\
    & (sec) & (kilobytes) & (sec) & (kilobytes)\\
    \hline
\multicolumn{5}{|c|}{\textit{$minsupp_{u}$} = $3$, \textit{$minsupp_{t}$} = $2$,}  \\
\multicolumn{5}{|c|}{\textit{$minsupp_{r}$}  = $1$, \textit{$minsupp_{d}$} = $1$} \\
\hline
\hline
                       $40000$ &      $\textbf{0. 05}$ & $\textbf{114}$    & $7.13$ & $309453$\\

                 $80000$ &       $\textbf{0. 10}$ & $\textbf{342}$    & $28. 12$ & $445431$\\

       $120000$ &      $\textbf{0. 22}$ & $\textbf{656}$    & $61. 60$ & $550932$\\

              $150000$ &      $\textbf{0. 45}$ & $\textbf{1241}$    & $119.45$ & $678542$\\

               $186479$ &      $\textbf{0. 77}$ & $\textbf{1754}$    & $255. 71$ & $788021$\\
\hline
\hline

\multicolumn{5}{|c|}{\textit{$minsupp_{u}$} = $2$, \textit{$minsupp_{t}$} = $2$,}  \\
\multicolumn{5}{|c|}{\textit{$minsupp_{r}$}  = $2$, \textit{$minsupp_{d}$} = $1$} \\
\hline
\hline
                       $40000$ &      $\textbf{0. 39}$ & $\textbf{177}$    & $32. 29$ & $456323$\\

                 $80000$ &       $\textbf{0. 53}$ & $\textbf{421}$    & $57. 06$ & $590012$\\

       $120000$ &      $\textbf{1. 60}$ & $\textbf{782}$    & $182. 40$ & $698672$\\

              $150000$ &      $\textbf{3. 39}$ & $\textbf{1025}$    & $354. 71$ & $826862$\\

               $186479$ &      $\textbf{5. 87}$ & $\textbf{1672}$    & $496.55$ & $932871$\\
\hline
\hline

\multicolumn{5}{|c|}{\textit{$minsupp_{u}$} = $2$, \textit{$minsupp_{t}$} = $2$,}  \\
\multicolumn{5}{|c|}{\textit{$minsupp_{r}$}  = $1$, \textit{$minsupp_{d}$} = $1$} \\
\hline
\hline
                      $40000$ &      $\textbf{0. 84}$ & $\textbf{1876}$    & $51. 88$ & $498672$\\

                 $80000$ &       $\textbf{2. 94}$ & $\textbf{3891}$    & $201. 58$ & $780762$\\

       $120000$ &      $\textbf{8. 71}$ & $\textbf{6789}$    & $487. 92$ & $1198451$\\

              $150000$ &      $\textbf{17. 81}$ & $\textbf{11342}$    & $1049.34$ & $1343572$\\

               $186479$ &      $\textbf{29. 78}$ & $\textbf{14562}$    & $1949.14$ & $1552789$\\
\hline
\hline

\hline
\multicolumn{5}{|c|}{\textit{$minsupp_{u}$} = $2$, \textit{$minsupp_{t}$} = $1$,}  \\
\multicolumn{5}{|c|}{\textit{$minsupp_{r}$}  = $1$, \textit{$minsupp_{d}$} = $1$} \\
\hline
\hline
                  $40000$ &      $\textbf{2. 91}$ & $\textbf{6724}$   & $89. 77$ & $1008273$\\

                 $80000$ &       $\textbf{6. 87}$ & $\textbf{11562}$    & $221. 93$ & $1336451$\\

       $120000$ &      $\textbf{21. 87}$ & $\textbf{14345}$    & $724. 47$ & $1542006$\\

              $150000$ &      $\textbf{46. 52}$ & $\textbf{15623}$    & $1524. 76$ & $1772919$\\

               $186479$ &      $\textbf{88. 16}$ & $\textbf{18976}$    & $3118. 85$ & $2188452$\\
\hline
\hline

\end{tabular}
\end{center}
\caption{Performances of \textsc{QuadriCons} \emph{vs}. \textsc{Data-Peeler} above the \textsc{Last.fm} dataset.}\label{perflast}
\end{table}

\subsection{Compacity of Quadri-Concepts}
Figure \ref{nadtemps} shows the number of frequent quadri-concepts versus the number of frequent quadri-sets on both \textsc{MovieLens} and \textsc{Last.fm} datasets for the different values of quadruples. We observe that for both datasets, the number of frequent quadri-sets increase massively when the number of quadruples grows. Indeed, frequent quadri-concepts become more large, \emph{i.e.,} containing more users, tags, resources and dates. Thus, such concepts cause the steep increase of frequent quadri-sets. For both datasets, the frequent quadri-concepts represent until $3$. $68$ \% and $28$. $99$ \% of the number of frequent quadri-sets. Hence, computing frequent quadri-sets is a harder task than computing frequent quadri-concepts while providing the same information.

\begin{figure}
\begin{center}
\parbox{6cm}{\includegraphics[scale=0.55]{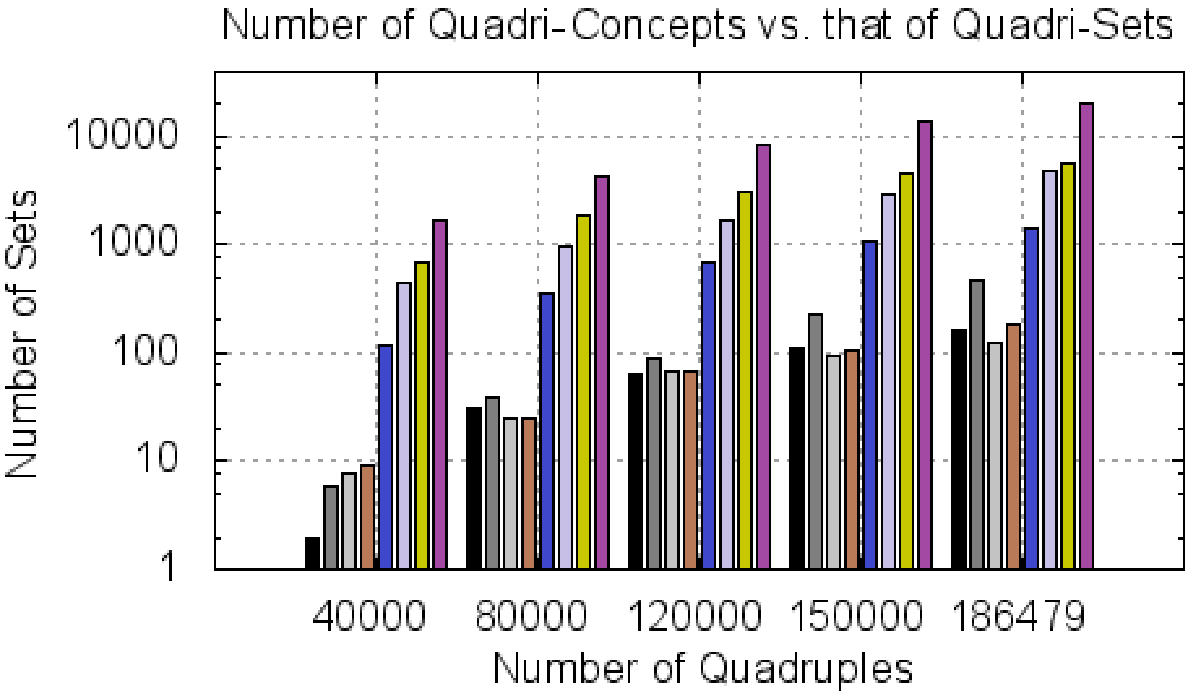}}
\parbox{6cm}{\includegraphics[scale=0.55]{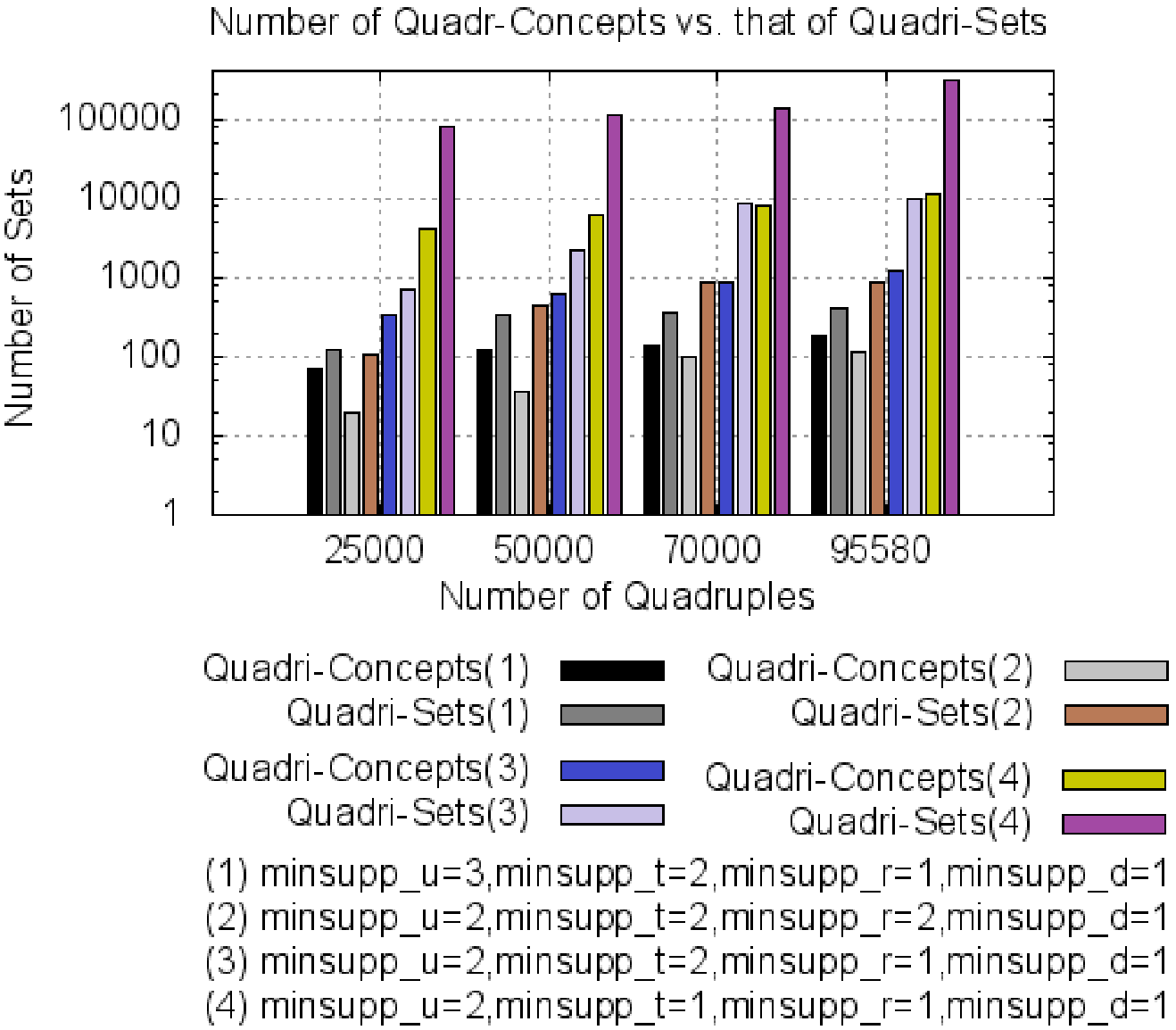}}

\caption{Number of frequent quadri-concepts \emph{vs.} number of frequent quadri-sets on both datasets. (\textbf{Top}) \textsc{Last.fm} (\textbf{Bottom}) \textsc{MovieLens}}
\label{nadtemps}
\end{center}
\end{figure}

\section{Conclusion and Perspectives}\label{sec7}
In this paper, we considered the quadratic context formally described by a \emph{d-folksonomy} with the introduction of a new dimension : time stamp. Indeed, we extend the notion of closure operator and tri-generator to the four-dimensional case and we thoroughly studied their theoretical properties. Then, we proposed the \textsc{Quadricons} algorithm in order to extract frequent quadri-concepts from \emph{d-folksonomies}. Several experiments show that \textsc{Quadricons} provides an efficient method for mining quadri-concepts in large scale conceptual structures. It is important to highlight that mining quadri-concepts stands at the crossroads of the avenues for future work : \emph{(i)} analyse evolution of users, tags and resources through time, \emph{(ii)} define the quadratic form of association rules according to quadri-concepts.

\bibliographystyle{IEEEtran}
\bibliography{biblio}

\end{document}